\documentclass[english]{revtex4-1}
\usepackage[T1]{fontenc}
\usepackage[latin9]{inputenc}
\usepackage{color}
\usepackage{array}
\usepackage{longtable}
\usepackage{textcomp}
\usepackage{amsthm}
\usepackage{amsmath}
\usepackage{amssymb}
\usepackage{stackrel}
\usepackage{graphicx}
\usepackage{esint}

\makeatletter

\providecommand{\tabularnewline}{\\}

\theoremstyle{plain}
\newtheorem{thm}{\protect\theoremname}
  \theoremstyle{plain}
  \newtheorem{lem}{\protect\lemmaname}
  \theoremstyle{plain}
  \newtheorem{cor}{\protect\corollaryname}

\makeatother

\usepackage{babel}
  \providecommand{\lemmaname}{Lemma}
\providecommand{\corollaryname}{Corollary}
\providecommand{\theoremname}{Theorem}

\begin{document}

\title{Exploring the ``Middle Earth'' of Network Spectra via a Gaussian
Matrix Function}

\author{Ernesto Estrada, Alhanouf Ali Alhomaidhi, Fawzi Al-Thukair}

\affiliation{Department of Mathematics \& Statistics, University of Strathclyde,
26 Richmond Street, Glasgow G11XQ, UK, Department of Mathematics,
King Saud University, Saudi Arabia}
\begin{abstract}
We study a Gaussian matrix function of the adjacency matrix of artificial
and real-world networks. In particular, we study the Gaussian Estrada
index---an index characterizing the importance of eigenvalues close
to zero. This index accounts for the information contained in the
eigenvalues close to zero in the spectra of networks. Here we obtain
bounds for this index in simple graphs, proving that it reaches its
maximum for star graphs followed by complete bipartite graphs. We
also obtain formulas for the Estrada Gaussian index of Erd\H{o}s-R�nyi
random graphs as well as for the Barab�si-Albert graphs. We also show
that in real-world networks this index is related to the existence
of important structural patterns, such as complete bipartite subgraphs
(bicliques). Such bicliques appear naturally in many real-world networks
as a consequence of the evolutionary processes giving rise to them.
In general, the Gaussian matrix function of the adjacency matrix of
networks characterizes important structural information not described
in previously used matrix functions of graphs.
\end{abstract}
\maketitle
\bigskip{}

\textbf{The spectrum of a network---the set of its eigenvalues---provides
important information about the structural and dynamical properties
of the corresponding system. Most of the functions used to study network
spectra give more weight to the largest modular eigenvalues. Then,
the information contained in the eigenvalues close to the centre of
the spectra, i.e, those close to zero, has remained totally unexplored
in the study of graph spectra. Here we study a Gaussian matrix function
that gives more weights to the eigenvalues closest to the centre of
the spectrum of a network. Using this function we extract important
structural information hidden in the spectra of networks, such as
emergence of complete bipartite subgraphs (bicliques) which appear
naturally in many real-world networks as a consequence of the evolutionary
processes giving rise to them. These bicliques are also ubiquitous
in random networks generated by preferential attachment mechanisms,
such as the Barab�si-Albert model. In this work we provide a series
of analytical results that pave the way for further analysis and uses
of this Gaussian matrix function to understand network structure and
dynamics.}

\pagebreak{}

\section{Introduction}

Matrix functions \cite{Function of matrices} have emerged as an important
mathematical tool for studying networks \cite{Estrada Higham}. The
concepts of communicability \cite{Communicability}, subgraph centrality
\cite{Protein folding,Subgraph Centrality} (see also \cite{Estrada Hatano Benzi}
for a review) and Katz index \cite{Katz centrality} are derived from
matrix functions $f\left(A\right)$ of the adjacency matrix and allow
the characterization of local structural properties of networks. The
trace of $f\left(A\right)$, which is known as the Estrada index of
the graph \cite{Estrada index,Estrada index 1,Estrada index 2}, is
a useful characterization of the global structure of a graph and it
has found applications as an index of natural connectivity for studying
robustness of networks \cite{Natural connectivity 1,Natural connectivity 2}.
These initial studies have motivated more recent developments in the
theory of graph-theoretic matrix function studies \cite{Benzi 1,Benzi 2,Benzi 3}.
All these indices have found multiple applications for studying real-world
social, ecological, biological, infrastructural, and technological
systems represented by networks \cite{Michaelbook,EstradaBook,LucianoReview}.
Here we will use interchangeably the terms networks and graphs and
will follow standard notation as in \cite{EstradaBook}. The greatest
appeal of the use of functions of the adjacency matrix for studying
graphs is that when representing them in terms of a Taylor function
expansion: $f\left(A\right)=\sum_{k=0}^{\infty}c_{k}A^{k}$, the entries
of the $k$th power of the adjacency matrix provides information about
the number of walks of length $k$ between the corresponding pair
of (not necessarily different) nodes (see next section for formal
definitions). Then, the important ingredient of the definition of
$f\left(A\right)$ lies in the use of the coefficients $c_{k}$. The
use of $c_{k}=k!^{-1}$ gives rise to the exponential function of
the adjacency matrix, which is the basis of the communicability/subgraph
centrality. On the other hand, selecting $c_{k}=\alpha^{-k}$ gives
rise to the resolvent of the adjacency matrix, which is the basis
of the Katz centrality index \cite{Katz centrality}. Either of these
two coefficients is selected arbitrarily among all the existing possibilities.
However, they have proved to be very useful in practice and not very
much improvement is obtained by changing the coefficients to account
for bigger or smaller penalization of the walks according to their
length \cite{Zooming in and out}.

Here we propose to investigate the information contained in the mid
part of the spectrum of the adjacency matrix of graphs and networks
using a new adjacency matrix function. The adjacency matrix of a simple
graph always contains positive and negative eigenvalues. Then, we
will refer here to the region close to the zero eigenvalue as the
middle part of the spectrum. This is only truly the middle part in
bipartite networks where the spectrum is symmetric, but we will use
the term without loss of generality for any graph. This region of
the spectrum is totally unexplored for complex networks. However,
there are areas in which the zero eigenvalue plays a fundamental role.
For instance, when the adjacency matrix represents the tight-binding
Hamiltonian in the H�ckel molecular orbital (HMO) method (see \cite{HMO_1,HMO_2}
for recent reviews), the zero eigenvalue and its multiplicity (graph
nullity) represent important parameters related to the molecular stability
and molecular magnetic properties (see \cite{Graph nullity review}
for a review). In these cases the highest occupied (HOMO) and lowest
unoccupied molecular orbitals (LUMO), which correspond to the smallest
positive and the smallest negative eigenvalue of $A$, respectively,
play the most fundamental role in the chemical reactivity. It can
be said that everything interesting in Chemistry takes place with
the involvement of the eigenvalues closest to zero. For instance,
many chemical reactions and electron transfer complexes involve electron
transfers between the HOMO of one molecule and the LUMO of another
\cite{Fukui,Fukui_2,Frontiers orbitals}.

Matrix functions of the type of $f\left(A\right)=\sum_{k=0}^{\infty}c_{k}A^{k}$
are characterized by the fact that they give the highest weight to
the largest eigenvalue of the adjacency matrix. For a simple example
let us consider the trace of $f\left(A\right)=\exp\left(A\right)$
of a simple, connected network, which can be written as $tr\exp\left(A\right)=\sum_{j=1}^{n}\exp\left(\lambda_{j}\right)$,
where $n$ is the order of the graph and $\lambda_{1}>\lambda_{2}\geq\cdots\geq\lambda_{n}$
are the eigenvalues of $A$. It is clear that if the spectral gap
of the adjacency matrix, $\lambda_{1}-\lambda_{2}$, is very large,
$tr\exp\left(A\right)$ depends only of the largest eigenvalue $\lambda_{1}$.
This is not a strange situation in real-world networks, where it is
typical to find very large spectral gaps for their adjacency matrix.
In these cases the use of functions of the type $f\left(A\right)$
makes that the structural information contained in the smaller eigenvalues
and eigenvectors of the adjacency matrix is not captured by the index.
A similar situation happens if we consider $f\left(-A\right)$ \cite{Negative temperature}.
In this case we give more weight to the smallest eigenvalue/eigenvector
of the adjacency matrix and the information contained in the largest
ones is again lost.

In this work we study a Gaussian adjacency matrix function $f\left(-A^{2}\right)$
as a way to characterize the structural information of graphs giving
more importance to the eigenvalues/eigenvectors in the middle part
of the graph spectrum. Similar Gaussian operators may arise in quantum
mechanics of many body systems \cite{Gaussian operator 1,Gaussian operator 2}
as well as as the electronic partition function in renormalized tight
binding Hamiltonians \cite{renormalized,Estrada Benzi Gaussian}.
We start by proving some elementary results for some of the indices
derived from $f\left(-A^{2}\right)$ for general graphs. In particular
we study here properties of $H=trf\left(-A^{2}\right)$. We show that
although the graph nullity---the multiplicity of the zero eigenvalue
of the adjacency matrix of the graph---plays an important role in
the values of this index, the $H$ index contains more structural
information than the graph nullity even for small simple graphs. We
then prove that among the graphs with $n$ nodes, the maximum of the
$H$ index is always obtained for the star graph followed by other
complete bipartite graphs. Then, we obtain analytic expressions for
this index in random graphs with Poisson and power-law degree distribution,
showing that the last ones always display larger values of the $H$
index than the first ones. \textcolor{black}{Finally, we study more
than 60 real-world networks representing a large variety of complex
systems. In this case we study the }$H$ index normalized by the network
size, $\hat{H}$. We found that the networks with the largest $\hat{H}$
index correspond to those having relatively large bicliques---complete
bipartite subgraphs, which can be created by different evolutionary
mechanisms depending on the kind of complex system considered. Although
there are important network characteristics influencing the $\hat{H}$
index, such as degree distribution and the degree assortativity, we
show here that they are not unique in determining the high values
of this index observed for certain networks. This new matrix function
for graphs and networks may represent an important addition to the
characterization of important properties of these systems which have
remained unexplored due to the lack of characterizations of the 'middle
region' of graph spectra.

\section{Preliminaries }

Let us introduce some definitions, notations, and properties associated
with networks to make this work self-contained. We will use interchangeably
the terms graphs and networks in this work. A \textit{graph} $\Gamma=(V,E)$
is defined by a set of $n$ nodes (vertices) $V$ and a set of $m$
edges $E=\{(u,v)|u,v\in V\}$ between the nodes. Here we will consider
simple graphs without multiple edges, self-loops and direction of
the edges. A \textit{walk} of length $k$ in $G$ is a set of nodes
$i_{1},i_{2},\ldots,i_{k},i_{k+1}$ such that for all $1\leq l\leq k$,
$(i_{l},i_{l+1})\in E$. A \textit{closed walk} is a walk for which
$i_{1}=i_{k+1}$. A \textit{path} is a walk with no repeated nodes. 

Let $A$ be the adjacency operator on $\ell_{2}(V)$, namely $(Af)(p)=\sum_{q:{\rm dist}(p,q)=1}f(q)$
. For simple finite graphs $A$ is the symmetric adjacency matrix
of the graph, which has entries

\[
a_{uv}=\left\{ \begin{array}{ll}
1 & \mbox{if }(u,v)\in E\\
0 & \mbox{otherwise}
\end{array}\right.\qquad\forall u,v\in V.
\]

The degree $k_{i}$ of the node $i$ is the number of edges incident
to it, equivalently $k_{i}=\sum_{j}a_{ij}$. In the particular case
of an undirected network as the ones studied here, the associated
adjacency matrix is symmetric, and thus its eigenvalues are real.
We label the eigenvalues of $A$ in non-increasing order: $\lambda_{1}>\lambda_{2}\geq\ldots\geq\lambda_{n}$.
Since $A$ is a real-valued, symmetric matrix, we can decompose $A$
into $A=U\Lambda U^{T}$ where $\Lambda$ is a diagonal matrix containing
the eigenvalues of $A$ and $U=[\mathbf{\overrightarrow{\psi}}_{1},\ldots,\mathbf{\overrightarrow{\psi}}_{n}]$
is orthogonal, where $\mathbf{\overrightarrow{\psi}}_{i}$ is an eigenvector
associated with $\lambda_{i}$. Because the graphs considered here
are connected, $A$ is irreducible and from the Perron-Frobenius theorem
we can deduce that $\lambda_{1}>\lambda_{2}$ and that the leading
eigenvector $\mathbf{q}_{1}$, which will be sometimes referred to
as the \textit{Perron vector}, can be chosen such that its components
$\mathbf{\mathbf{\overrightarrow{\psi}}}_{1}(u)$ are positive for
all $u\in V$. 

Hereafter we will refer to the following function as the communicability
function of the graph \cite{Communicability,Estrada Higham,Estrada Hatano Benzi}.
Let $u$ and $v$ be two nodes of $\Gamma$. The communicability function
between these two nodes is defined as 

\[
G_{uv}=\sum_{k=0}^{\infty}\frac{\left(A^{k}\right)_{uv}}{k!}=\left(\exp\left(A\right)\right)_{uv}=\sum_{k=1}^{n}e^{\lambda_{k}}\mathbf{\mathbf{\overrightarrow{\psi}}}_{k}(u)\mathbf{\overrightarrow{\psi}}_{k}(v),
\]

which is an important quantity for studying communication processes
in networks. It counts the total number of walks starting at node
$u$ and ending at node $v$, weighted in decreasing order of their
length by a factor $\frac{1}{k!}$; therefore it is considering shorter
walks more influential than longer ones. The $G_{uu}$ terms of the
communicability function characterize the degree of participation
of a node in all subgraphs of the network, giving more weight to the
smaller ones. Thus, it is known as the subgraph centrality of the
corresponding node \cite{Subgraph Centrality}. The following quantity
is known in the algebraic graph theory literature as the Estrada index
of the graph:

\[
EE\left(G\right)=\sum_{u=1}^{n}G_{uu}=tr\left(\exp\left(A\right)\right)=\sum_{k=1}^{n}e^{\lambda_{k}},
\]

which is a characterization of the global properties of a network.
In its generalized form $EE\left(G,\beta\right)=tr\left(\exp\left(\beta A\right)\right)=\sum_{k=1}^{n}e^{\beta\lambda_{k}},$
it represents the statistical-mechanics partition function of the
graph where $\beta$ represents the inverse temperature.

In some parts of this work we will consider the following integral

\begin{equation}
I_{\gamma}(x)=\frac{1}{\pi}\stackrel[0]{\pi}{\int}\cos(\gamma\theta)\exp(x\cos\theta)d\theta-\frac{\sin(\gamma\pi)}{\pi}\stackrel[0]{\infty}{\int}\exp(-x\cosh t-\gamma t)dt,
\end{equation}

which corresponds to the modified Bessel function of the first kind.

\section{Gaussian Adjacency Matrix Function of Networks}

With the goal of accounting for the influence of the eigenvalues close
to middle of the spectrum of the adjacency matrix of a graph, i.e.,
those close to zero, we start here by introducing the following matrix
function

\begin{equation}
\tilde{G}=\sum_{k=0}^{\infty}\dfrac{\left(-A^{2}\right)^{k}}{k!}=\exp(-A^{2}).
\end{equation}

By obvious reasons we will call it the \textit{Gaussian matrix function}
of $A$. Let $\tilde{G}_{pq}$ be the \textit{Gaussian communicability
function} between the nodes $p$ and $q$ based on $-A^{2}$. That
is,
\begin{equation}
\tilde{G}_{pq}=(\exp(-A^{2}))_{pq}
\end{equation}

Correspondingly, $\tilde{G}_{pp}$ is the \textit{Gaussian subgraph
centrality} based on the same matrix function. The trace of $\exp(-A^{2})$
will be designated by 
\begin{equation}
H=tr(\exp(-A^{2})),
\end{equation}

which corresponds to the \textit{Gaussian Estrada index} of the graph.
It is very important to mention that for calculating the index $H$
we do not need to obtain explicitly the exponential matrix of $-A^{2}$.
We are not interested here in the development of such kind of techniques
but the reader is directed to the excellent work of Benzi and Boito
\cite{Benzi_Boito} for a discussion of efficient� techniques for
estimating the trace of an exponential matrix that do not require
computing every entry of the matrix exponential.� 

Obviously, using the spectral decomposition of the adjacency matrix
we can express both indices as

\begin{equation}
\tilde{G}_{pq}=\sum_{j=1}^{n}\psi_{j,p}\psi_{j,q}\exp\left(-\lambda_{j}^{2}\right),
\end{equation}

\begin{equation}
H=\sum_{j=1}^{n}\exp\left(-\lambda_{j}^{2}\right).
\end{equation}

Let $\eta\left(A\right)$ be the nullity of the adjacency matrix $A$,
i.e., the dimension of the null space of $A$. In spectral graph theory
$\eta=\eta\left(A\right)$ is known as the \textit{graph nullity}.
Then, it is obvious that the $H$ index is related to $\eta$ as follows:

\begin{equation}
H\geq\eta,
\end{equation}

with both indices identical if and only if $\lambda_{j}=0$, for all
$j$, which is attained only for the trivial graph, i.e., the graph
with $n$ nodes and no edges. Indeed, 

\begin{equation}
H=\eta+\sum_{\lambda_{j}\neq0}\exp\left(-\lambda_{j}^{2}\right).\label{eq:H-nullity}
\end{equation}

Then, it is interesting at least empirically, to explore the relation
between $H$ and $\eta$ for simple graphs. We investigate all the
connected graphs with $n\leq8$ for which we obtain both $H$ and
$\eta$. The correlation between both indices for the 11,117 connected
graphs with 8 nodes is illustrated in Figure (\ref{general correlations}).
Although the correlation is statistically significant---the Pearson
correlation coefficient is 0.74---it hides the important differences
between the two indices. For instance, there are 5,724 graphs with
zero nullity among all the connected graphs with 8 nodes. For these
graphs $1.484\leq H\leq3.629$, which represents a wide range of values
taking into account that the minimum and maximum values of $H$ for
all connected graphs with 8 nodes are 1.484 and 6, respectively. It
is also easy to see that there are graphs having nullity zero which
have larger $H$ indices than some graphs having nullity one, two
or three. The results are very similar for $n<8$ and they are not
shown here. In the Figure (\ref{general correlations}) we show the
graphs with the largest $H$ indices among all connected graphs with
8 nodes and nullity zero or one. These graphs show a common pattern
containing several complete bipartite subgraphs. For instance, every
yellow node in the Figure (\ref{general correlations}) is connected
to every red ones, every red is connected to every blue and every
blue is connected to the green one, while there is no yellow-yellow,
red-red or blue-blue connections. This pattern will be revealed when
we study the mathematical properties of this index and its importance
will be analyzed for real-world networks.

\begin{figure}
\begin{centering}
(a)\includegraphics[width=0.33\textwidth]{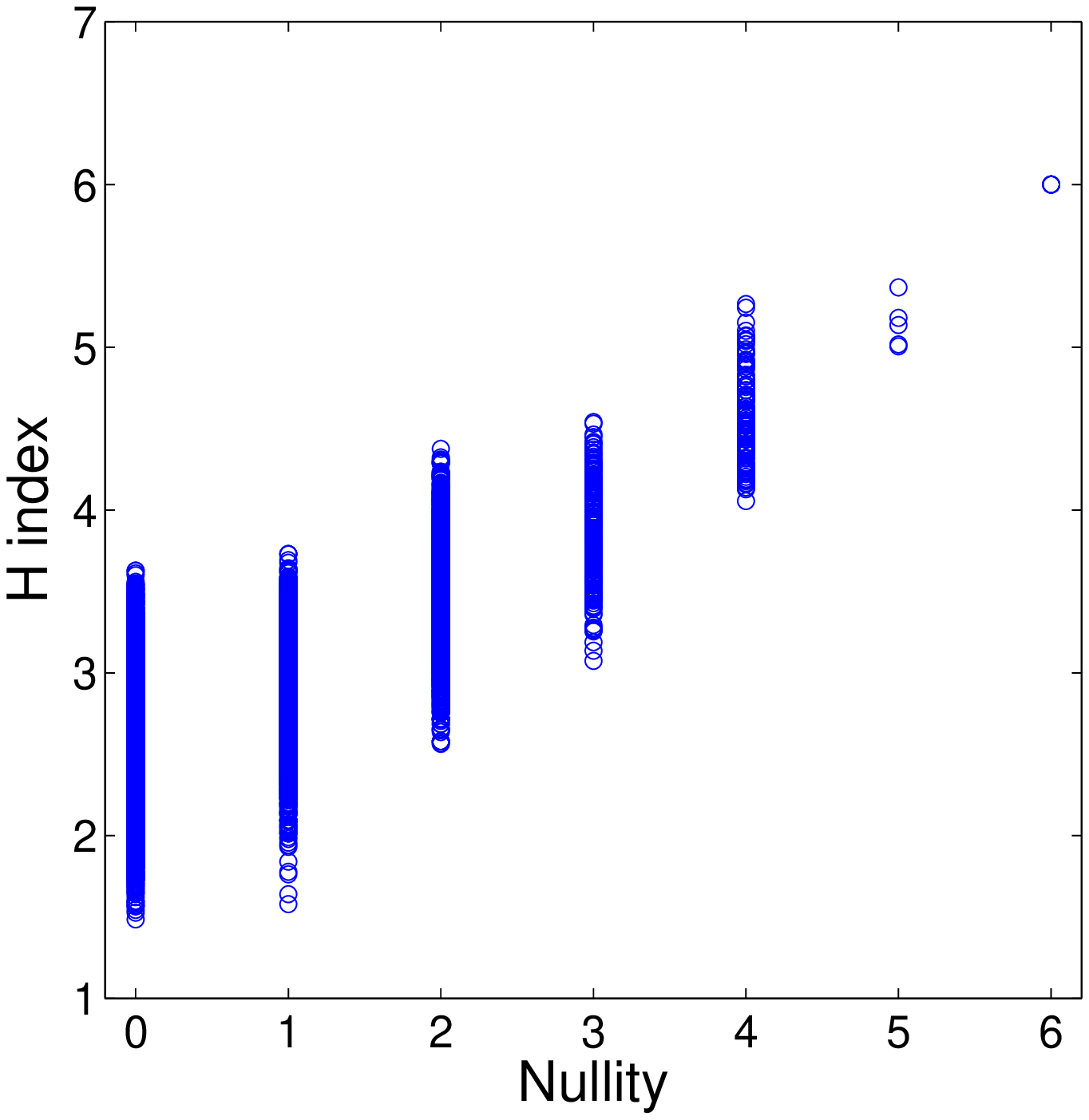}\quad{}(b)\includegraphics[width=0.25\textwidth]{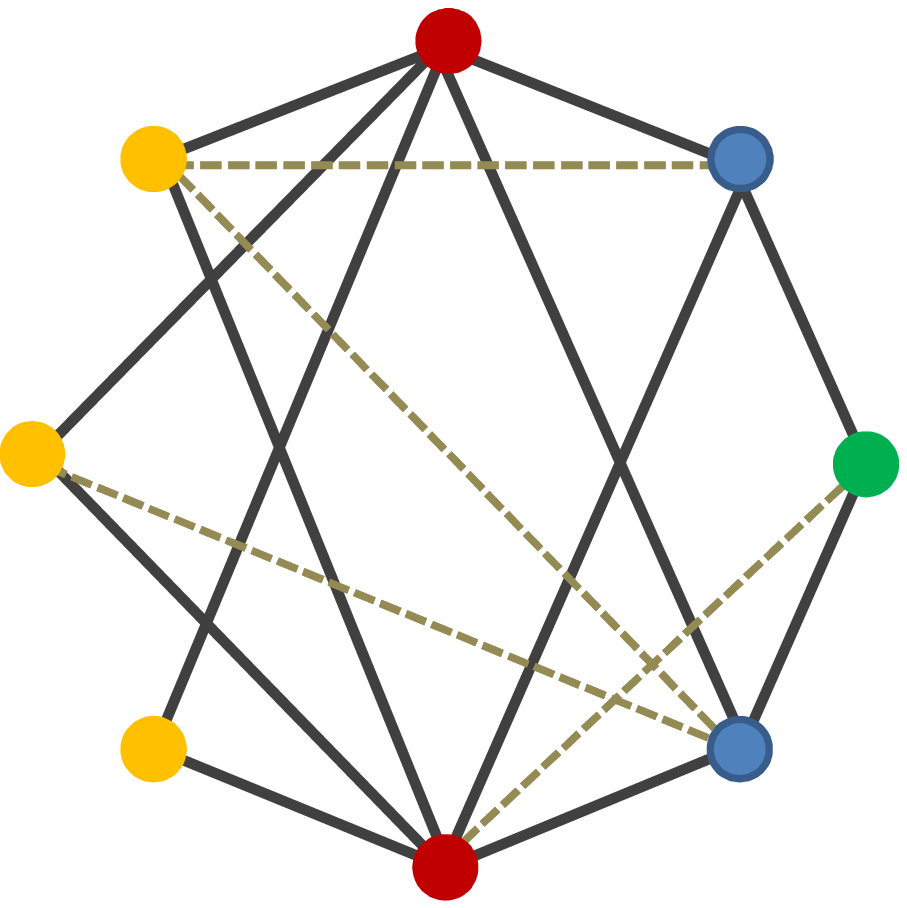}\quad{}(c)\includegraphics[width=0.25\textwidth]{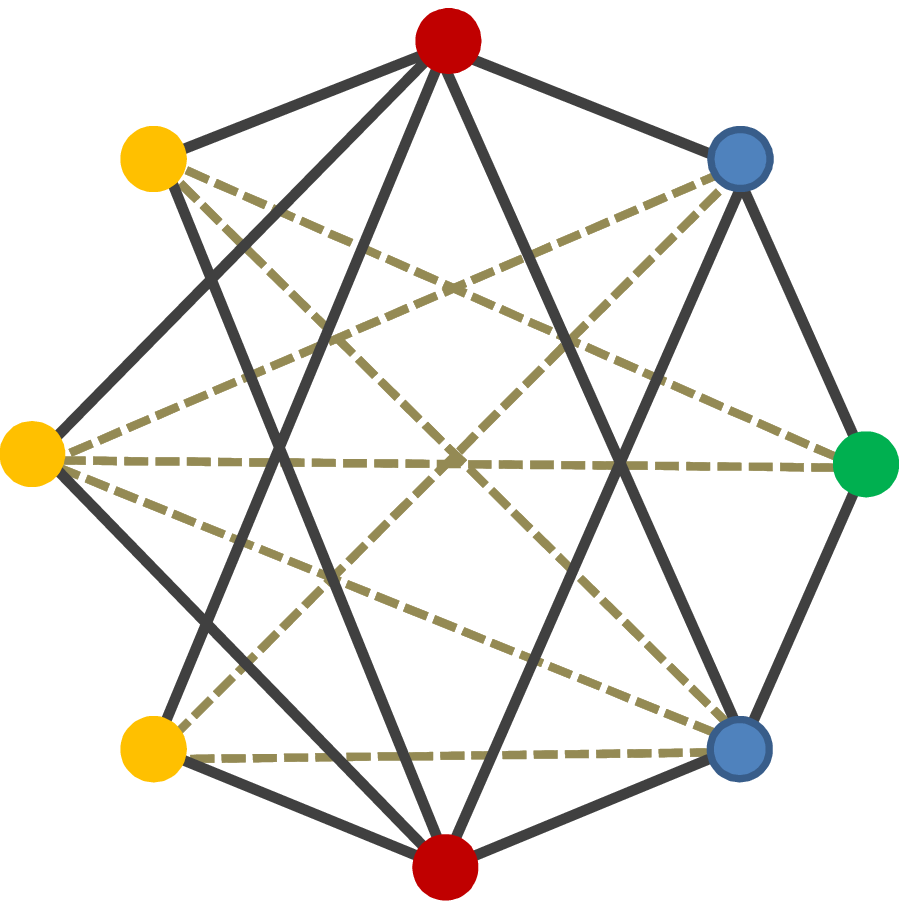}
\par\end{centering}

\protect\caption{(a) Plot of graph nullity versus $H$ index for all connected graphs
with 8 nodes. (b) Graph with the largest $H$ index among all the
connected graphs with 8 nodes and nullity zero. (c) The same as in
(b) for all connected graphs with nullity one. Notice in (b) and (c)
the connectivity pattern of the graphs in which every yellow node
is connected to every red ones, every red is connected to every blue
and every blue is connected to the green one. Also, there is no yellow-yellow,
red-red or blue-blue connections. }

\label{general correlations}
\end{figure}

\subsection{General Quadrature Rule-Based Bounds }

In this section we will use the Gaussian quadrature rule to obtain
an upper bound of $H$. We will mainly follow here the works \cite{quadrature rule 1,quadrature rule 2,quadrature rule 3}
to which the reader is directed for more. We start by recalling that
for a symmetric matrix $A$ and a smooth function $f$ defined on
an interval containing the eigenvalues of $A,\left[a,b\right]$ we
have:

\begin{equation}
I\left[f\right]=u^{T}f(A)v=\stackrel[b]{a}{\int}f(\lambda)d\mu(\lambda)\textnormal{ where }\mu(\lambda)=\begin{cases}
0, & \lambda\lneq a=\lambda_{1}\\
\stackrel[j=1]{i}{\sum}p_{j}q_{j,} & \lambda_{i}\leq\lambda\lneq\lambda_{i+1}\\
\stackrel[j=1]{N}{\sum}p_{j}q_{j,} & b=\lambda_{n}\leq\lambda.
\end{cases}
\end{equation}

Our motivation for using this definition is the fact that $\left[f(A)\right]_{ij}=e_{i}^{T}f(A)e_{j}$,
where $e_{i}$ is the $i$th column of the identity matrix. Moreover,
$u^{T}f(A)v=\stackrel[b]{a}{\int}f(\lambda)d\mu(\lambda)=\stackrel[j=1]{n}{\sum}w_{j}f(t_{j})+\stackrel[k=1]{M}{\sum}v_{k}f(z_{k})+R\left[f\right]$
which is the general Gauss-type quadrature rule where the nodes $\left\{ t_{j}\right\} _{j=1}^{n}$
and wights $\left\{ w_{j}\right\} _{j=1}^{n}$ are unknowns, whereas
the nodes $\left\{ z_{k}\right\} _{k=1}^{M}$ are prescribed . We
have:
\begin{itemize}
\item $M=0$ for the Gauss rule,
\item $M=1$, $z_{1}=a$ or $z_{1}=b$ for the Gauss-Radau rule,
\item $M=2$, $z_{1}=a$ and $z_{2}=b$ for the Gauss-Lobatto rule, which
we will focus on.
\end{itemize}
Let $J_{n}$ be a tridiagonal matrix defined as 

\[
J_{n}=\left[\begin{array}{ccccc}
\omega_{1} & \gamma_{1}\\
\gamma_{1} & \omega_{2} & \gamma_{2}\\
 & \ddots & \ddots & \ddots\\
 &  & \gamma_{n-2} & \omega_{n-1} & \gamma_{n-1}\\
 &  &  & \gamma_{n-1} & \omega_{n}
\end{array}\right],
\]
whose eigenvalues are the Gauss nodes, whereas the Gauss wights are
given by the square of the first entries of the normalized eigenvectors
of $J_{n}$, then,

\begin{equation}
\stackrel[l=1]{N}{\sum}w_{l}f(t_{l})=e_{1}^{T}f(J_{n})e_{1}.
\end{equation}

The entries of $J_{n}$ are computed using the symmetric Lanczos algorithm.
Now, if $f$ is a strictly completely monotonic function on an interval
$I=[a,b]$ containing the eigenvalues of the matrix $A$, i.e $f^{(2j)}(x)>0$
and $f^{(2j+1)}(x)<0$ on $I$ for all $j\geq0$ where $f^{(k)}$
denotes the $k$th derivative of $f$ and $f^{(0)}\equiv f$, the
symmetric Lanczos process can be used to compute bounds for the diagonal
entries $(f(A))_{ii}$. Let $J_{2}=\left[\begin{array}{cc}
\omega_{1} & \gamma_{1}\\
\gamma_{1} & \omega_{2}
\end{array}\right]$ be the Jacobian matrix obtained by taking a single Lanczos step,
then we only need to compute the $(1,1)$ entry of $f(J_{2})$. Now,
if $\varphi(x,y)=\frac{\omega_{1}\left(f(x)-f(y)\right)+xf(y)-yf(x)}{x-y}$,
then the Gauss-Lobatto rule gives the bound 
\[
(f(A))_{ii}\leq\varphi(a,b).
\]
See \cite{quadrature rule 1} for more details.

The main result of this section is the following.
\begin{thm}
Let $G$ be a graph with $n$ nodes and $m$ edges and let $H=tr\exp\left(-A^{2}\right)$.
Then,
\end{thm}
\begin{equation}
H\left(G\right)\leq\stackrel[i=1]{n}{\sum}\left[\frac{d_{i}(e^{-b}-1)}{b}+1\right]=2m\frac{(e^{-b}-1)}{b}+n.
\end{equation}

\begin{proof}
In the case of $J_{2}$ we have:

\[
J_{2}=\left[\begin{array}{cc}
\omega_{1} & \gamma_{1}\\
\gamma_{1} & \omega_{2}
\end{array}\right],\omega_{1}=a_{ii},\gamma_{1}^{2}=\underset{i\neq j}{\sum}a_{ij}^{2},\omega_{2}=\frac{1}{\gamma_{1}^{2}}\underset{k\neq i}{\sum}\underset{l\neq i}{\sum}a_{ki}a_{kl}a_{li}
\]
Now, if $B=A^{2}$, where $A$ is the adjacency matrix of a graph
$G$ and $f(x)=e^{-x}$ we have

\[
J_{2}=\left[\begin{array}{cc}
d_{i} & \sqrt{\underset{i\neq j}{\sum}b_{ij}^{2}}\\
\sqrt{\underset{i\neq j}{\sum}b_{ij}^{2}} & \frac{1}{\sqrt{\underset{i\neq j}{\sum}b_{ij}^{2}}}\underset{k\neq i}{\sum}\underset{l\neq i}{\sum}b_{ki}b_{kl}b_{li}
\end{array}\right]
\]

where $d_{i}$ is the degree of the node $i$ and $b_{ij}$ is the
$(i,j)$th entry of $A^{2}$. Notice that $b_{ij}=\stackrel[k=1]{n}{\sum}a_{ik}a_{kj}$
and $\left[a,b\right]=\left[0,b\right]$ since $A^{2}$ has nonnegative
eigenvalues.

Hence, we have for the Gauss-Lobatto rule
\[
(e^{-A^{2}})_{ii}\leq\frac{d_{i}(1-e^{-b})-b}{-b}=\frac{d_{i}(e^{-b}-1)}{b}+1,
\]

To find the bound of the trace of $e^{-A^{2}}$ we take the summation
from $1$ to $n$ on the previous inequality, which by the Handshaking
Lemma gives the final result.
\end{proof}

\section{H Index of Graphs}

\subsection{Elementary properties}

In the following we show some results about $\tilde{G}_{pq}$ of some
elementary graphs which will help us to interpret this measure when
applied to more complex structures. In particular, we study the $n$-nodes
path $P_{n}$, the $n$-nodes cycle $C_{n}$, the star graph $K_{1,n-1}$,
the complete graph $K_{n}$ of $n$ nodes and the complete bipartite
graph $K_{n_{1},n_{2}}$ of $n_{1}+n_{2}$ nodes. $P_{n}$ is a connected
graph in which $n-2$ nodes are connected to other two nodes and two
nodes are connected to only one node; $C_{n}$ is the connected graph
of $n$ nodes in which every node is connected to two others; $K_{1,n-1}$is
the connected graph in which there is one node connected to $n-1$
nodes, here labeled as $1$ and named the central node, and $n-1$
nodes are connected to the central one only; $K_{n}$ is the graph
in which every pair of nodes is connected by an edge; and $K_{n_{1},n_{2}}$
is the connected graph which is formed by two sets $V_{1}$ and $V_{2}$
of nodes of cardinalities $n_{1}$ and $n_{2}$, respectively, such
that every node in $V_{1}$ is connected to every node in $V_{2}$.
Here we give expressions for the $H\left(G\right)$ index of the before
mentioned graphs in the form of Lemmas.
\begin{lem}
Let $K_{n}$ be the complete graph of $n$ nodes. Then 
\begin{equation}
H\left(K_{n}\right)=e^{-(n-1)^{2}}+\frac{n-1}{e}.
\end{equation}

\label{a}\end{lem}
\begin{proof}
The spectrum of $K_{n}$ is $\sigma(K_{n})=\left\{ \left[n-1\right]^{1},\left[-1\right]^{n-1}\right\} $
with the eigenvector $\varphi_{1}=\frac{1}{\sqrt{n}}(1,1,\ldots,1)$
so we have
\begin{equation}
\tilde{G}_{pq}\left(K_{n}\right)=\varphi_{1}(p)\varphi_{1}(q)e^{-(n-1)^{2}}+\stackrel[j=2]{n}{\sum}\varphi_{j}(p)\varphi_{j}(q)e^{-1}
\end{equation}

And since the eigenvector matrix has orthonormal rows and columns
we have $\stackrel[j=2]{n}{\sum}\varphi_{j}(p)\varphi_{j}(q)=-\frac{1}{n}$
if $p\ne q$ and $\frac{n-1}{n}$ if $p=q$.
\begin{equation}
\tilde{G}_{pq}\left(K_{n}\right)=\frac{e^{-(n-1)^{2}}}{n}-\frac{1}{ne}
\end{equation}

Now, if $p=q$ then $\tilde{G}_{pp}\left(K_{n}\right)=\varphi_{1}^{2}(p)e^{-(n-1)^{2}}+\stackrel[j=2]{n}{\sum}\varphi_{j}^{2}(p)e^{-1}=\frac{e^{-(n-1)^{2}}}{n}+\frac{n-1}{ne}$.

Then, it is straightforward to realize that
\begin{eqnarray}
H & \left(K_{n}\right)= & \stackrel[j=1]{n}{\sum}(\frac{e^{-(n-1)^{2}}}{n}+\frac{n-1}{ne})\\
 & = & e^{-(n-1)^{2}}+\frac{n-1}{e}
\end{eqnarray}

\end{proof}

\begin{lem}
Let $P_{n}$ be a path having $n$ nodes. Then, asymptotically as
$n\rightarrow\infty$ 
\begin{equation}
H\left(P_{n}\right)=\frac{I_{0}(2)}{e^{2}}(n+1)-e^{-4}.
\end{equation}
\end{lem}
\begin{proof}
By substituting the eigenvalues and eigenvectors of the path graph
into the expression for $\tilde{G}_{pp}\left(P_{n}\right)$ we obtain
\begin{eqnarray}
\tilde{G}_{pp}\left(P_{n}\right) & = & \frac{2}{n+1}\stackrel[j=1]{n}{\sum}\sin^{2}\left(\frac{j\pi p}{n+1}\right)\exp\left(-4\cos^{2}\left(\frac{j\pi}{n+1}\right)\right)\\
 & = & \frac{e^{-2}}{n+1}\stackrel[j=1]{n}{\sum}\left[1-\cos\left(\frac{2j\pi p}{n+1}\right)\right]\exp\left(-2\cos\left(\frac{2j\pi}{n+1}\right)\right).\label{eqb}
\end{eqnarray}
 Now, when $n\rightarrow\infty$ the summation in \ref{eqb} can be
approached by the following integral 
\begin{equation}
\tilde{G}_{pp}\left(P_{n}\right)=\frac{e^{-2}}{\pi}\int_{0}^{\pi}\exp(-2\cos\theta)d\theta-\frac{e^{-2}}{\pi}\int_{0}^{\pi}\cos\left(p\theta\right)\exp(-2\cos\theta)d\theta,
\end{equation}
 where $\theta=\frac{2j\pi}{n+1}$. Thus, when $n\rightarrow\infty$
we have 
\begin{equation}
\tilde{G}_{pp}\left(P_{n}\right)=e^{-2}\left(I_{0}(-2)-I_{p}(-2)\right)
\end{equation}
 which by using $I_{\gamma}(-x)=(-1)^{\gamma}I_{\gamma}(x)$ gives
\begin{eqnarray*}
\tilde{G}_{pp} & \left(P_{n}\right)= & e^{-2}\left(I_{0}(2)-(-1)^{p}I_{p}(2)\right).
\end{eqnarray*}

Let $n$ be even. Then due to the symmetry of the path we have 
\begin{eqnarray}
H\left(P_{n}\right) & = & 2\stackrel[p=1]{n/2}{\sum}\tilde{G}_{pp}\left(P_{n}\right)=2\stackrel[p=1]{n/2}{\sum}e^{-2}\left[I_{0}(2)-(-1)^{p}I_{p}(2)\right]\\
 & = & \frac{nI_{0}(2)}{e^{2}}-\frac{2}{e^{2}}\stackrel[p=1]{n/2}{\sum}(-1)^{p}I_{p}(2).
\end{eqnarray}
 For $n\rightarrow\infty$ we have 
\begin{equation}
\stackrel[\gamma=1]{\infty}{\sum}(-1)^{\gamma}I_{\gamma}(x)=\frac{1}{2}\left(e^{-x}-I_{0}(x)\right).
\end{equation}
 Then, we can write for $n\rightarrow\infty$
\begin{eqnarray}
H\left(P_{n}\right) & = & \frac{nI_{0}(2)}{e^{2}}-\frac{1}{e^{2}}\left(e^{-2}-I_{0}(2)\right)\\
 & = & \frac{I_{0}(2)}{e^{2}}(n+1)-e^{-4}.
\end{eqnarray}
Now, when $n$ is odd we can split the path into two paths of lengths
$\frac{n+1}{2}$ and $\frac{n-1}{2}$, respectively. Then, we write
\begin{eqnarray}
H\left(P_{n}\right) & = & \stackrel[p=1]{\frac{n+1}{2}}{\sum}\tilde{G}_{pp}\left(P_{n}\right)+\stackrel[p=\frac{n-1}{2}]{n}{\sum}\tilde{G}_{pp}\left(P_{n}\right)\\
 & = & \frac{(n+1)I_{0}(2)}{2e^{2}}-\frac{1}{e^{2}}\stackrel[p=1]{\frac{n+1}{2}}{\sum}(-1)^{p}I_{p}(2)+\frac{(n-1)I_{0}(2)}{2e^{2}}-\frac{1}{e^{2}}\stackrel[p=\frac{n-1}{2}]{n}{\sum}(-1)^{p}I_{p}(2).\label{eqc}
\end{eqnarray}
When $n\rightarrow\infty$ we can consider that the summation in the
second and fourth terms of \ref{eqc} are both equal to $\left(e^{-2}-I_{0}(2)\right)/2$,
which then gives the final result.
\end{proof}

\begin{lem}
Let $C_{n}$ be a cycle having $n$ nodes. Then, asymptotically as
$n\rightarrow\infty$
\begin{equation}
H(C_{n})=\frac{nI_{0}(-2)}{e^{2}}.
\end{equation}
\end{lem}
\begin{proof}
(\textbf{Lemma 3}): Notice that the adjacency matrix of a cycle is
a circulant matrix and consequently any function of it and that gives
\begin{eqnarray}
H(C_{n}) & = & \stackrel[j=1]{n}{\sum}\tilde{G}_{pp},\,for\,any\,node\,p\\
 & = & n\left(\frac{tr(e^{-A^{2}})}{n}\right)\\
 & = & n\left(\frac{1}{n}\stackrel[j=1]{n}{\sum}e^{-4\cos^{2}(\frac{2\pi j}{n})}\right)\label{eq:c}\\
 & = & ne^{-2}\left(\stackrel[j=1]{n}{\sum}\frac{1}{n}e^{-2\cos\frac{4\pi j}{n}}\right)\label{eqd}
\end{eqnarray}
Now, when $n\rightarrow\infty$ the summation in \ref{eqd} can be
approached by the following integral
\begin{equation}
H(C_{n})=ne^{-2}\frac{1}{\pi}\int_{0}^{\pi}e^{-2cos\theta}d\theta
\end{equation}
where $\theta=\frac{2j\pi}{n}$. Thus, when $n\rightarrow\infty$
we have
\begin{equation}
H(C_{n})=ne^{-2}I_{0}(-2).
\end{equation}

\end{proof}

\begin{lem}
Let $K_{n_{1},n_{2}}$ be the complete bipartite graph of $n_{1}+n_{2}$
nodes. Then
\begin{equation}
H\left(K_{n_{1},n_{2}}\right)=2e^{-n_{1}n_{2}}+n_{1}+n_{2}-2.
\end{equation}

\end{lem}
The following corollary will be of importance in the following section
of this work.
\begin{proof}
(\textbf{Lemma 4}): From the orthonormality of the eigenvectors of
the adjacency matrix we have:
\begin{equation}
\stackrel[j=2]{n_{1}+n_{2}-1}{\sum}\left[\varphi_{j}\left(p\right)\right]^{2}=1-\frac{1}{n_{1}},\,p\in V_{1}
\end{equation}
\begin{equation}
\stackrel[j=2]{n_{1}+n_{2}-1}{\sum}\left[\varphi_{j}\left(p\right)\right]^{2}=1-\frac{1}{n_{2}},\,p\in V_{2}
\end{equation}

Hence, if $p\in V_{1}$
\begin{eqnarray}
\tilde{G}_{pp}\left(K_{n_{1},n_{2}}\right) & = & \stackrel[j=1]{n_{1}+n_{2}}{\sum}\left[\varphi_{j}\left(p\right)\right]^{2}\exp(-\lambda_{j}^{2})\\
 & = & e^{-n_{1}n_{2}}(\frac{n_{1}n_{2}}{2n_{1}n_{2}^{2}}+\frac{n_{1}n_{2}}{2n_{2}n_{1}^{2}})+\stackrel[j=2]{n_{1}+n_{2}-1}{\sum}\left[\varphi_{j}\left(p\right)\right]^{2}\\
 & = & e^{-n_{1}n_{2}}(\frac{1}{n_{1}})+1-\frac{1}{n_{1}}=\frac{1}{n_{1}}(e^{-n_{1}n_{2}}-1)+1,
\end{eqnarray}

and similarly we have $\tilde{G}_{pp}\left(K_{n_{1},n_{2}}\right)=\frac{1}{n_{2}}(e^{-n_{1}n_{2}}-1)+1$
when $p\in V_{2}$. Then
\begin{eqnarray}
H & \left(K_{n_{1},n_{2}}\right)= & \stackrel[j=1]{n_{1}+n_{2}}{\sum}\tilde{G}_{pp}\\
 & = & \stackrel[j=1]{n_{1}}{\sum}\tilde{G}_{pp}+\stackrel[j=m+1]{n_{1}+n_{2}}{\sum}\tilde{G}_{pp}\\
 & = & n_{1}(\frac{1}{n_{1}}(e^{-n_{1}n_{2}}-1)+1)+n_{2}(\frac{1}{n_{2}}(e^{-n_{1}n_{2}}-1)+1)\\
 & = & 2e^{-n_{1}n_{2}}+n_{1}+n_{2}-2.
\end{eqnarray}

\end{proof}

\begin{cor}
Let $K_{1,n-1}$ be the star graph of $n$ nodes. Then
\begin{equation}
H\left(K_{1,n-1}\right)=2e^{1-n}+n-2.
\end{equation}

\end{cor}

\subsection{Graphs with maximum H index}

Here we are mainly interested in understanding why certain networks
display large values of the $H$ index. Then, we prove that among
the graphs with $n$ nodes, the maximum value of the $H$ index is
always obtained for the star graph $K_{1,n-1}$. We start this section
by proving a general results for trees, which is needed to prove the
upper bound.
\begin{lem}
\label{c}Let $T_{n}$ be a tree of $n$ nodes, then
\end{lem}
\begin{equation}
H(T_{n})\leq H(K_{1,n-1})
\end{equation}

\begin{proof}
We have the following upper bound

\begin{equation}
H(G)\leq2m\frac{(e^{-b}-1)}{b}+n\label{a-1}
\end{equation}
 where $m$ is the number of edges and $[0,b]$ is the interval that
contains all the eigenvalues of $A^{2}$. Since $A$ is irreducible
then it has a nonnegative real eigenvalue (name it $\lambda_{1}$)
which has maximum absolute value among all eigenvalues (Perron-Frobenius).

Now, Let $T_{n}$ be a tree with $n\geq2$, then Collatz and Sinogowitz
\cite{Collatz and Sinogowitz} have proved that 

\begin{equation}
\lambda_{1}(T_{n})\leq\lambda_{1}(K_{1,n-1})=\sqrt{n-1},
\end{equation}
 where the equality holds if $T_{n}$ is the star graph. Thus, the
interval $[0,n-1]$ contains all the eigenvalues of any tree $T_{n}$.
Now, substituting in (\ref{a-1})
\begin{equation}
H(T_{n})\leq2(n-1)\frac{(e^{1-n}-1)}{n-1}+n=n-2+2e^{1-n}.\label{b-1}
\end{equation}

Thus, for any tree of $n$ nodes $H(T_{n})\leq H(K_{1,n-1})$ .
\end{proof}
Now we prove an important result for general graphs, which also allow
us to understand the nature of the index $H$ when studying real-world
networks.
\begin{thm}
Let G be connected graph of n nodes, then

\begin{equation}
H(G)\leq H(K_{1,n-1}).
\end{equation}
\end{thm}
\begin{proof}
The largest eigenvalue of any graph $G$ is less than or equal the
maximum degree. Thus the interval $[0,(n-1)^{2}]$ contains all the
eigenvalues of $A^{2}$ and we get from the quadrature-rule bound

\begin{equation}
H(G)\leq n-2m\frac{(1-e^{-(n-1)^{2}})}{(n-1)^{2}}.
\end{equation}
Now, $H(G)$ is maximum when $m$ is the lowest possible for a connected
graph. That is,
\begin{equation}
H(G)\leq n-2\frac{(1-e^{-(n-1)^{2}})}{n-1}.
\end{equation}
A connected graph with $n-1$ edges is a tree. Then, because of Lemma
(\ref{c}) we have that
\begin{equation}
H(G)\leq H(K_{1,n-1}).
\end{equation}

\end{proof}
Obviously, when $n\rightarrow\infty$, $H\left(K_{1,n-1}\right)\rightarrow n-2.$
In a similar way, when $n\rightarrow\infty$

\begin{equation}
H\left(K_{n_{1},n_{2}}\right)\rightarrow n_{1}+n_{2}-2=n-2.
\end{equation}

Thus last expression indicates that the complete bipartite graphs
also display the largest value of the $H$ index asymptotically as
$n\rightarrow\infty$. Indeed, we have studied all connected graphs
with 5, 6, 7, and 8 nodes and observed the following. Among the graphs
with $n$ nodes, as proved here, the maximum value is always reached
for the star graph $K_{1,n-1}$. It is then followed by the complete
bipartite graph $K_{2,n-2}$, then $K_{3,n-3}$, and so forth. For
instance, in the case $n=8$ we have $H\left(K_{1,7}\right)\approx6.001824$;
$H\left(K_{2,6}\right)\approx6.000012$; $H\left(K_{3,5}\right)\approx6.000001$;
$H\left(K_{4,4}\right)\approx6.000000$. This observation will play
a fundamental role in the analysis of random graphs and real-world
networks in the next sections of this work.

\subsection{Graphs with minimum H index }

As we have seen before (see Eq. (\ref{eq:H-nullity})) the largest
contribution to the $H$ index is made by the graph nullity $\eta$
and by the eigenvalues which are relatively close to zero. Let $x>0$
be a real number such that $\exp\left(-x^{2}\right)\sim0$ . Then, 

\begin{equation}
H\approx\sum_{\lambda_{j}\geq-x}^{\lambda_{j}\leq x}\exp\left(-\lambda_{j}^{2}\right).
\end{equation}

Consequently, the graphs with minimum $H$ index are those having
very small density of eigenvalues in the interval $\left(-x,x\right)$.
For instance, the graph having the smallest $H$ index among all connected
graphs with 8 nodes has eigenvalues: -2.0000, -1.7321, -1.0000, -1.0000,
-0.8136, 1.4707, 1.7321, 3.3429, which produces $H\approx1.4845$,
which is well approximated if we consider only the eigenvalues in
the interval $\left(-1.5,1.5\right)$. The graphs with minimum $H$
index among all connected graphs with $n=4,5,6,7,8$ are illustrated
in the Figure \ref{Minimum H}. A complete structural characterization
of these graphs is out of the scope of this work, but it calls the
attention the existence of bow-tie subgraphs in most of these graphs.

\begin{figure}
\begin{centering}
\includegraphics[width=0.75\textwidth]{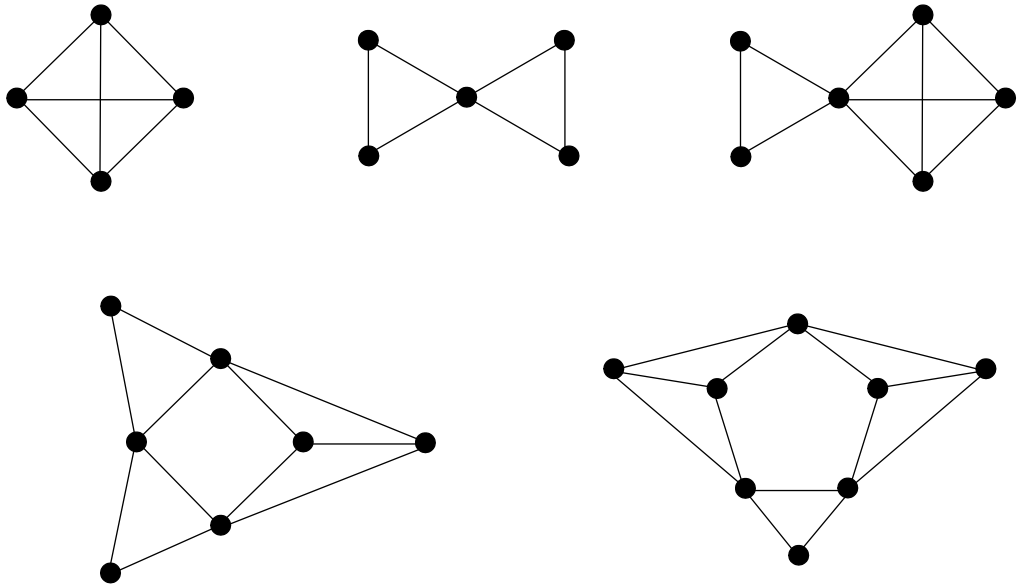}
\par\end{centering}

\protect\caption{Illustration of the graphs having minimum $H$ index among all connected
graphs with $n=4,5,6,7,8$. }

\label{Minimum H}
\end{figure}

\subsection{H Index of Random Networks}

In this section we study two different models of random graphs. They
are very ubiquitous as null models for studying real-world networks.
The first model is the Erd\H{o}s-R�nyi $G\left(n,p\right)$ \cite{ER model}
also known as the Gilbert model \cite{Gilbert}, in which a graph
with $n$ nodes is constructed by connecting nodes randomly in such
a way that each edge is included in $G\left(n,p\right)$ with probability
$p$ independent from every other edge. The second model was introduced
by Barab�si and Albert \cite{BA model} on the basis of a preferential
attachment process. In this model the graph is constructed from an
initial seed of $m_{0}$ vertices connected randomly like in an Erd\H{o}s-R�nyi
$G\left(n,p\right)$. Then, new nodes are added to the network in
such a way that each new node is connected to the existing ones with
a probability that is proportional to the degree of these existing
nodes. While the Erd\H{o}s-R�nyi $G\left(n,p\right)$ random graphs
have a Poisson degree distribution (when $n\rightarrow\infty$), the
Barab�si-Albert ones show power-law degree distribution of the form:
$p\left(k\right)\sim k^{-3},$where $p\left(k\right)$ is the probability
of finding a node with degree $k$. In term of their spectra the main
difference is that ER graphs display the Wigner semi-circle distribution
\cite{Wigner} of eigenvalues when $n\rightarrow\infty$ of the form

\begin{equation}
\mathrm{\rho(\lambda)=\begin{cases}
\frac{2\sqrt{r^{2}-\lambda^{2}}}{\pi r^{2}} & ,-r\leq\lambda\leq r\\
0, & otherwise,
\end{cases}}\label{eq:semicircle}
\end{equation}

where $r=2\sqrt{np(1-p)}$ . However, the BA networks have a triangular
distribution \cite{EstradaBook} of eigenvalues of the form

\begin{equation}
\rho(\lambda)=\begin{cases}
\frac{\lambda+r}{r^{2}}, & -r\leq\lambda<0\\
\frac{r-\lambda}{r^{2}}, & 0<\lambda\leq r\\
0, & otherwise.
\end{cases}\label{eq:triangular}
\end{equation}

Using these distributions we obtain the following results.
\begin{thm}
For an Erd\H{o}s-R�nyi random graph $G(n,p)$ with $\frac{\ln n}{n}\ll p$
we have
\begin{equation}
H\left(ER\right)=ne^{\frac{-r^{2}}{2}}(I_{0}(\frac{r^{2}}{2})+I_{1}(\frac{r^{2}}{2}))\label{eq:H(ER)}
\end{equation}

almost surely, as $n\rightarrow\infty$, where $r=2\sqrt{np(1-p)}$
and $I_{n}$ is the modified Bessel function of the first kind.\end{thm}
\begin{proof}
We know that the spectral density of $G(n,p)$ converges to the semicircular
distribution (\ref{eq:semicircle}) as $n\rightarrow\infty$. Also,
Krivelevich and Sudakov \cite{Krivelevich and Sudakov} showed that
the largest eigenvalue $\lambda_{1}$ of $G(n,p)$ is almost surely
$(1+o(1))np$ provided that $np\gg\ln n$. Then, 
\begin{eqnarray}
H\left(ER\right) & = & \exp(-\lambda_{1}^{2})+\stackrel[i=2]{n}{\sum}\exp(-\lambda_{i}^{2})\\
 & = & e^{-\lambda_{1}^{2}}+n\left(\frac{1}{n}\stackrel[i=2]{n}{\sum}e^{-\lambda_{i}^{2}}\rho(\lambda)\right)
\end{eqnarray}

When $n\rightarrow\infty$ we have
\begin{eqnarray}
H\left(ER\right) & = & n\int_{-r}^{r}\rho(\lambda)e^{-\lambda^{2}}d\lambda\\
 & = & \frac{4n}{\pi r^{2}}\int_{0}^{r}\sqrt{r^{2}-\lambda^{2}}e^{-\lambda^{2}}d\lambda\\
 & = & \frac{4n}{\pi r^{2}}\int_{0}^{\frac{\pi}{2}}r^{2}\cos^{2}\theta e^{-r^{2}\sin^{2}\theta}d\theta\\
 & = & \frac{4n}{\pi}\int_{0}^{\frac{\pi}{2}}\frac{1}{2}(1+\cos2\theta)e^{\frac{-r^{2}}{2}(1-\cos2\theta)}d\theta\\
 & = & 2ne^{\frac{-r^{2}}{2}}(\frac{1}{\pi}\int_{0}^{\frac{\pi}{2}}e^{\frac{r^{2}}{2}\cos2\theta}d\theta+\frac{1}{\pi}\int_{0}^{\frac{\pi}{2}}\cos2\theta e^{\frac{r^{2}}{2}\cos2\theta}d\theta)\\
 & = & ne^{\frac{-r^{2}}{2}}(\frac{1}{\pi}\int_{0}^{\pi}e^{\frac{r^{2}}{2}\cos u}du+\frac{1}{\pi}\int_{0}^{\pi}\cos ue^{\frac{r^{2}}{2}\cos u}du)\\
 & = & ne^{\frac{-r^{2}}{2}}(I_{0}(\frac{r^{2}}{2})+I_{1}(\frac{r^{2}}{2}))
\end{eqnarray}

\end{proof}
We now consider the case of the Barab�si-Albert (BA) model as a representative
of random graphs with power-law degree distribution. We then prove
the following result.
\begin{thm}
Let $G$ be a $BA$ random network. Then, when $n\rightarrow\infty$,
the $H$ index of a BA network is bounded as
\begin{equation}
H\left(BA\right)=\dfrac{n}{r^{2}}\left(\sqrt{\pi}r\textnormal{erf}\left(r\right)+e^{-r^{2}}-1\right).\label{eq:H(BA)}
\end{equation}

where $r=2\sqrt{np(1-p)}$ and $\textnormal{erf}\left(\right)$is
the error function.\end{thm}
\begin{proof}
We know that the density of $BA$ graphs follows a triangular distribution
(\ref{eq:triangular}). Thus
\begin{eqnarray}
H & \left(BA\right)= & \stackrel[j=1]{n}{\sum}\rho(\lambda_{j})e^{-\lambda_{j}^{2}}\\
 & = & n\left(\frac{1}{n}\stackrel[j=1]{n}{\sum}\rho(\lambda_{j})e^{-\lambda_{j}^{2}}\right)\\
 & = & n\left(\stackrel[-r]{r}{\int}\rho(\lambda)e^{-\lambda^{2}}d\lambda,\,\,as\,n\rightarrow\infty\right)\\
 & = & n\left(\stackrel[-r]{0}{\int}\frac{\lambda+r}{r^{2}}e^{-\lambda^{2}}d\lambda+\stackrel[0]{r}{\int}\frac{r-\lambda}{r^{2}}e^{-\lambda^{2}}d\lambda\right)\\
 & _{=} & \dfrac{n}{r^{2}}\left(\sqrt{\pi}r\textnormal{erf}\left(r\right)+e^{-r^{2}}-1\right).
\end{eqnarray}

\end{proof}
In Figure (\ref{H index ER and BA}(a)) we illustrate the results
obtained for the $H$ index of ER random graphs $G_{ER}\left(1000,p\right)$
in which $p$ is systematically changed from 0.008 to 0.04. The results
are shown for both, the formula (\ref{eq:H(ER)}) and the calculation
using the function 'expm' implemented in Matlab\textregistered . As
can be seen for ER networks, as soon as the probability increases,
such that $np\gg\ln n$, the two results quickly converge to a common
value, i.e., the error decay quickly with the increase of $p$. In
Figure (\ref{H index ER and BA}(b)) we also plot similar results
for the BA model using $G_{BA}\left(1000,m_{0}\right)$ in which $m_{0}$
is systematically varied from 4 to 20. In this case the behavior is
more complex as there is a crossing point between the two curves.
This difference between the behavior of the theoretical function (\ref{eq:H(BA)})
for low and large densities of the graphs may be due to the fact that
the eigenvalue distribution of the BA networks is different at these
two density regimes. According to our computational experiments, it
is only true that the BA networks display triangular eigenvalue distributions
for relatively small edge densities and deformations of it occurs
for larger densities, which may produce the observed deviations from
the theoretical and computational results. More theoretical work is
needed to understand completely the eigenvalue distribution of these
networks at different density regimes. Such studies are clearly out
of the scope of the current work.

\begin{figure}
\begin{centering}
(a)\includegraphics[width=0.5\textwidth]{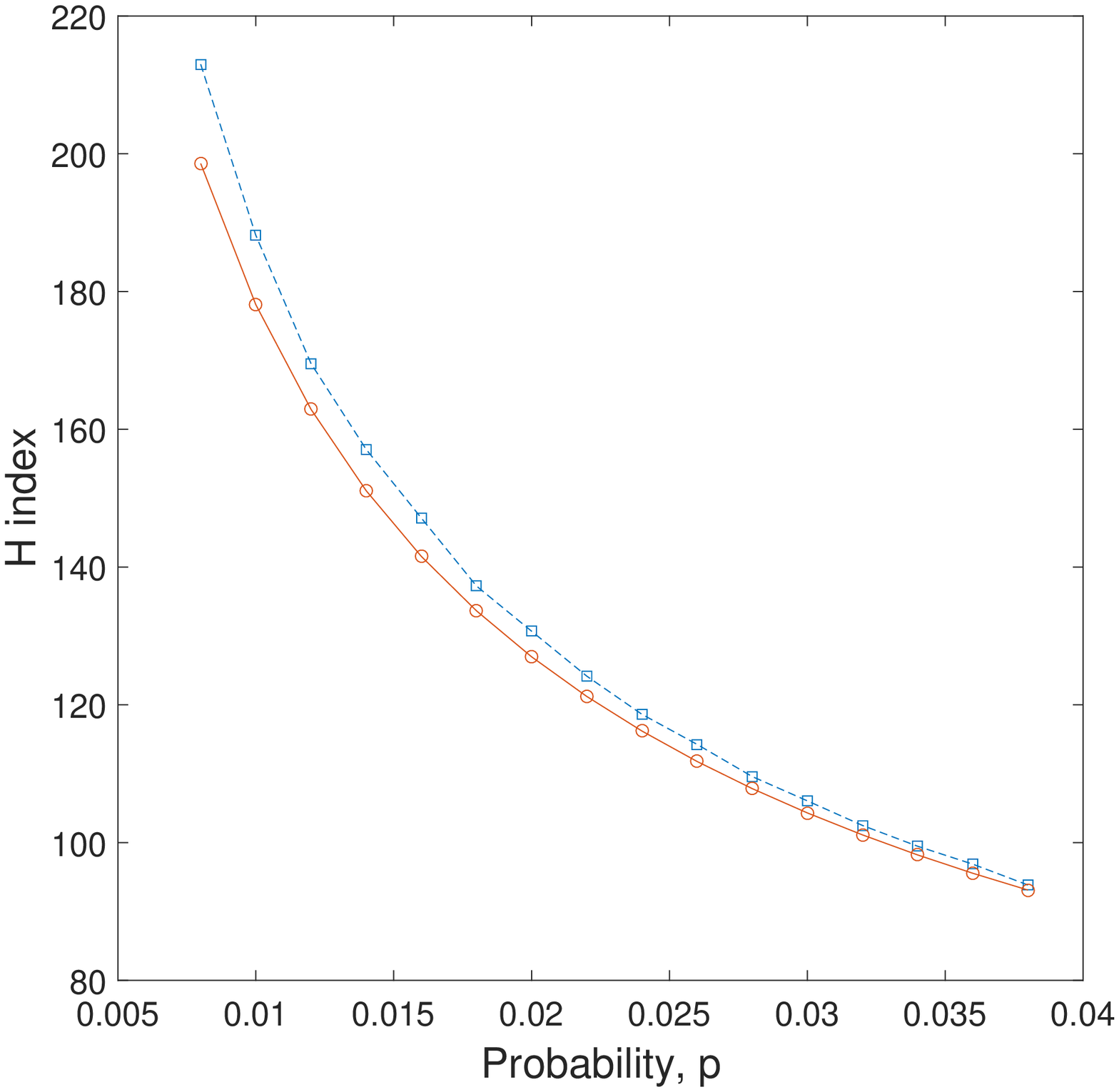} \\
(b)\includegraphics[width=0.5\textwidth]{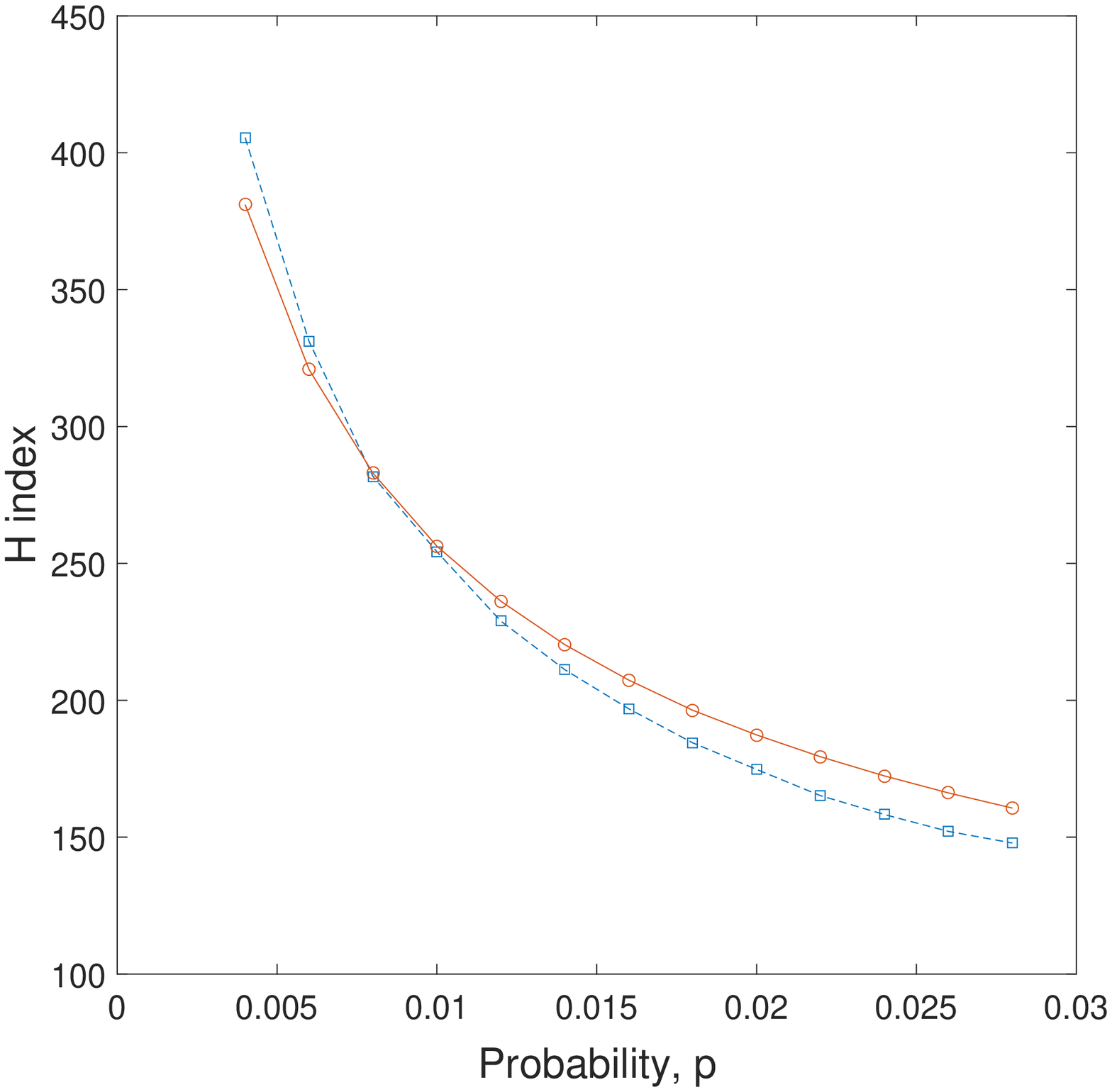}
\par\end{centering}

\protect\caption{(a) Change of the $H$ index with the increase of the probability
$p$ in ER random graphs $G_{ER}\left(1000,p\right)$ obtained using
the formula (\ref{eq:H(ER)}) (empty circles and solid line) and using
the function 'expm' in Matlab (squares and broken line). (b) Change
of the $H$ index with the increase of $m_{0}$ in BA random graphs
$G_{BA}\left(1000,m_{0}\right)$ obtained using the formula (\ref{eq:H(BA)})
(empty circles and solid line) and using the function 'expm' in Matlab
(squares and broken line). All the calculations are the average of
100 random realizations. }

\label{H index ER and BA}
\end{figure}

It is easy to show that for a given value of $r$, $H\left(BA\right)>H\left(ER\right).$
That is, for the same network density the network having power-law
degree distribution has larger value of the $H$ index than the analogous
one with Poisson degree distribution. This result is somehow expected
from the qualitative analysis of the eigenvalues distributions of
these two classes of random networks. While the ER networks display
a semicircle distribution of eigenvalues, the BA networks for small
values of $r$ displays a triangular distribution peaked at $\lambda_{j}=0$.
In other words, the nullity of the BA graphs is larger than that of
the ER ones, and the concentration of eigenvalues close to zero is
also larger for the BA networks than for the ER. Both characteristics
give rise to larger values of the $H$ index in the BA networks. The
question that arises here is what this difference implies from the
structural point of view. We will analyze this question in the remaining
part of this section.

We have already seen that the largest values of the $H$ index occurs
in graphs having complete bipartite structures. Then, in order to
understand the main structural differences giving rise to the larger
$H$ index in BA networks than in ER ones we consider the existence
of such subgraphs in both networks. In particular, we will consider
the existence of complete bipartite subgraphs, known as bicliques,
in both kind of networks. In the current work we will give only a
qualitative explanation of this difference which will point to the
direction of a further quantitative analysis. Let us start by the
analysis of the BA networks. These networks are created from an initial
seed of $n_{0}$ nodes connected randomly and independently according
to the ER model. Then, at each stage of the evolution of the network,
a new node is connected preferentially to $m_{0}\leq n_{0}$ nodes.
The connection probability is proportional to the degree of the existing
nodes. Because an ER network is uncorrelated the probability that
the highest degree nodes are connected to each other is relatively
low. Then, when a new node is added and connected to $m_{0}$ of the
highest degree existing nodes there is a high probability that a biclique
is formed. Such a process is continued as more nodes are added to
the graph, resulting a large bicliques with high probability (see
Figure \ref{BA model}). The creation of an ER network follows a completely
different process in which pairs of nodes are connected randomly and
independently, which does not generates any preferred subgraphs, thus
not producing a large number of bicliques. This qualitative analysis
explaining structurally the existence of networks with high values
of the $H$ index will be very useful in the next section of this
work where we will analyze real-world networks.

\begin{figure}
\begin{centering}
\includegraphics[width=1\textwidth]{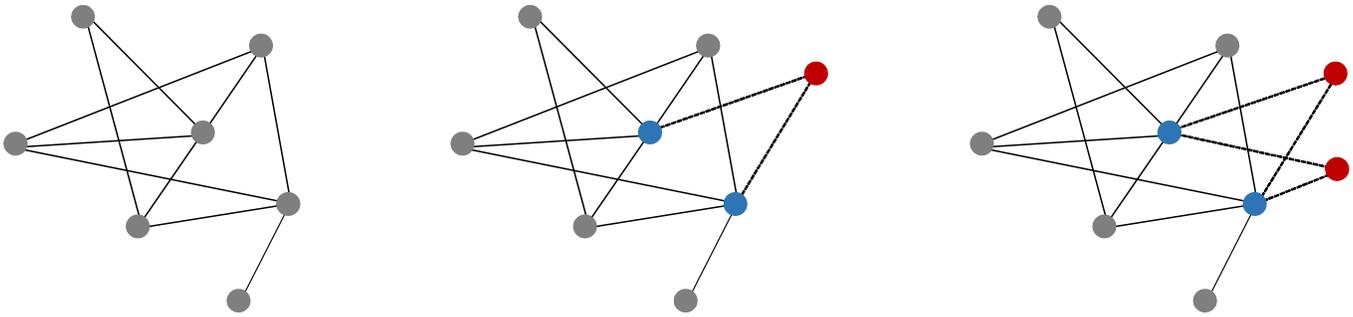}
\par\end{centering}

\protect\caption{Illustration of the evolution of a graph under the BA model to sketch
how bicliques are formed in such kind of networks. (a) Seed of $n_{0}=7$
nodes created with a Poissonian degree distribution to start the BA
evolution process. (b) Given $m_{0}=2$ the new node (red one) is
preferentially attached to those with the highest degree among the
existing $n_{0}$ ones (marked in blue). (c) Second iteration of the
process, which creates a biclique $K_{2,2}$ (red and blue nodes joined
by dotted lines).}

\label{BA model}
\end{figure}

\section{Studies of real-world networks }

\subsection{Datasets}

In this section we study a group of real-world networks representing
a variety of social, environmental, technological, infrastructural
and biological complex systems. A description of the networks and
their main characteristics are given below. 

\textit{Brain networks }
\begin{itemize}
\item Neurons: Neuronal synaptic network of the nematode \textit{C. elegans}.
Included all data except muscle cells and using all synaptic connections
\cite{Milo datasets}; Cat and macaque visual cortices: the brain
networks of macaque visual cortex and cat cortex, after the modifications
introduced by Sporn and K�tter \cite{Brain networks datasets}.
\end{itemize}
\textit{Ecological networks }
\begin{itemize}
\item Benguela: Marine ecosystem of Benguela off the southwest coast of
South Africa \cite{Benguela}; Bridge Brook: Pelagic species from
the largest of a set of 50 New York Adirondack lake food webs \cite{Bridge Brooks};
Canton Creek: Primarily invertebrates and algae in a tributary, surrounded
by pasture, of the Taieri River in the South Island of New Zealand
\cite{Canton}; Chesapeake Bay: The pelagic portion of an eastern
U.S. estuary, with an emphasis on larger fishes \cite{Chesapeake};
Coachella: Wide range of highly aggregated taxa from the Coachella
Valley desert in southern California \cite{Coachella}; El Verde:
Insects, spiders, birds, reptiles and amphibians in a rainforest in
Puerto Rico \cite{ElVerde}; Grassland: all vascular plants and all
insects and trophic interactions found inside stems of plants collected
from 24 sites distributed within England and Wales \cite{Grassland};
Little Rock: Pelagic and benthic species, particularly fishes, zooplankton,
macroinvertebrates, and algae of the Little Rock Lake, Wisconsin,
U.S. \cite{LittleRock}; Reef Small: Caribbean coral reef ecosystem
from the Puerto Rico-Virgin Island shelf complex \cite{Reef Small};
Scotch Broom: Trophic interactions between the herbivores, parasitoids,
predators and pathogens associated with broom, Cytisus scoparius,
collected in Silwood Park, Berkshire, England, UK \cite{Scotch Broom};
Shelf: Marine ecosystem on the northeast US shelf \cite{Shelf}; Skipwith:
Invertebrates in an English pond \cite{Skipwith}; St. Marks: Mostly
macroinvertebrates, fishes, and birds associated with an estuarine
seagrass community, Halodule wrightii, at St. Marks Refuge in Florida
\cite{StMarks}; St. Martin: Birds and predators and arthropod prey
of Anolis lizards on the island of St. Martin, which is located in
the northern Lesser Antilles \cite{StMartin}; Stony Stream: Primarily
invertebrates and algae in a tributary, surrounded by pasture, of
the Taieri River in the South Island of New Zealand in native tussock
habitat \cite{Stony}; Ythan\_1: Mostly birds, fishes, invertebrates,
and metazoan parasites in a Scottish Estuary \cite{Ythan1} ;Ythan\_2:
Reduced version of Ythan1 with no parasites \cite{Ythan2}. 
\item Termite: The networks of three-dimensional galleries in termite nests
\cite{Termite Moulds}; Ant: The network of galleries created by ants
\cite{Ants}; Dolphins: social network of frequent association between
62 bottlenose dolphins living in the waters off New Zealand \cite{Dolphins}; 
\end{itemize}
\textit{Informational networks }
\begin{itemize}
\item Centrality: Citation network of papers published in the field of Network
Centrality \cite{centrality,deNooy}; GD: Citation network of papers
published in the Proceedings of Graph Drawing during the period 1994-2000
\cite{GD}; ODLIS: Vocabulary network of words related by their definitions
in the Online Dictionary of Library and Information Science. Two words
are connected if one is used in the definition of the other \cite{ODLIS};
Roget: Vocabulary network of words related by their definitions in
Roget\textquoteright s Thesaurus of English. Two words are connected
if one is used in the definition of the other \cite{Roget}; Small
World: Citation network of papers that cite S. Milgram's 1967 Psychology
Today paper or use Small World in title \cite{Batagelj Mrvar}. 
\end{itemize}
\textit{Biological networks }
\begin{itemize}
\item Protein-protein interaction networks in: \textit{Kaposi sarcoma herpes
virus} (KSHV) \cite{KSHV}; \textit{P. falciparum} (malaria parasite)
\cite{malaria}; \textit{S. cerevisiae} (yeast) \cite{PIN_yeast_1,PIN_yeast_2};
\textit{A. fulgidus} \cite{PIN_Afulgidus}; \textit{H. pylori} \cite{PIN_Hpylori};
\textit{E. coli} \cite{PIN_Ecoli} and \textit{B. subtilis} \cite{PIN_Bsubtilis}.
\item Trans\_E.coli: Direct transcriptional regulation between operons in
\textit{Escherichia coli} \cite{MiloDatasets2,MiloDatasets3}; Trans\_sea\_urchin:
Developmental transcription network for sea urchin endomesoderm development.
\cite{MiloDatasets2}; Trans\_yeast: Direct transcriptional regulation
between genes in \textit{Saccaromyces cerevisae}. \cite{Milo datasets,MiloDatasets2}. 
\end{itemize}
\textit{Social and economic networks }
\begin{itemize}
\item Corporate: American corporate elite formed by the directors of the
625 largest corporations that reported the compositions of their boards
selected from the Fortune 1000 in 1999 \cite{Corporate}; Geom: Collaboration
network of scientists in the field of Computational Geometry \cite{Batagelj Mrvar};
Prison: Social network of inmates in prison who chose \textquotedblleft What
fellows on the tier are you closest friends with?\textquotedblright{}
\cite{Prison}; Drugs: Social network of injecting drug users (IDUs)
that have shared a needle in the last six months \cite{Moody}; Zachary:
Social network of friendship between members of the Zachary karate
club \cite{Zachary}; College: Social network among college students
in a course about leadership. The students choose which three members
they wanted to have in a committee \cite{College}; ColoSpring: The
risk network of persons with HIV infection during its early epidemic
phase in Colorado Spring, USA, using analysis of community wide HIV/AIDS
contact tracing records (sexual and injecting drugs partners) from
1985-1999 \cite{ColoSpg}; Galesburg: Friendship ties among 31 physicians
\cite{deNooy}; High\_Tech: Friendship ties among the employees in
a small high-tech computer firm which sells, installs, and maintain
computer systems \cite{HighTech,deNooy}; Saw Mills: Social communication
network within a sawmill, where employees were asked to indicate the
frequency with which they discussed work matters with each of their
colleagues \cite{SawMill,deNooy}; 
\end{itemize}
\textit{Technological and infrastructural networks }
\begin{itemize}
\item Electronic: Three electronic sequential logic circuits parsed from
the ISCAS89 benchmark set, where nodes represent logic gates and flip-flop
\cite{Milo datasets}; USAir97: Airport transportation network between
airports in US in 1997 \cite{Batagelj Mrvar}; Internet: The internet
at the Autonomous System (AS) level as of September 1997 and of April
1998 \cite{Internet}; Power Grid: The power grid network of the Western
USA \cite{WattsStrogatz}. 
\end{itemize}
\textit{Software networks }
\begin{itemize}
\item Collaboration networks associated with six different open-source software
systems, which include collaboration graphs for three Object Oriented
systems written in C++, and call graphs for three procedural systems
written in C. The class collaboration graphs are from version 4.0
of the VTK visualization library; the CVS snapshot dated 4/3/2002
of Digital Material (DM), a library for atomistic simulation of materials;
and version 1.0.2 of the AbiWord word processing program. The call
graphs are from version 3.23.32 of the MySQL relational database system,
and version 1.2.7 of the XMMS multimedia system. Details of the construction
and/or origin of these networks are provided in Myers \cite{Myers Software}.
\end{itemize}

\subsection{Analysis of real-world networks}

The sizes of the networks studied here range from 29 to 4,941 nodes.
Then, in order to avoid any size influence, we normalize the $H$
index by dividing it by the number of nodes of the network. We will
call $\hat{H}$ to the normalized index. The normalized index $\hat{H}$
ranges from about 0.14 to about 0.75 for the studied networks, indicating
that real-world networks cover most of the values that this index
can take (see \ref{Table-1.-Dataset}). The scatterplot of the normalized
nullity versus the normalized $H$ index for the 61 real-world networks
studied here (plot not shown) reveals that although both indices follow
the same trend, there are important differences among them. In particular,
we can observe that there are 9 networks with zero nullity which display
values of $\hat{H}$ ranging from about 0.14 (the lowest $\hat{H}$
index) to about 0.36 (ranked 25th in increasing order of $\hat{H}$
index). 

\begin{longtable}{>{\raggedright}p{3.5cm}>{\centering}p{1.2cm}>{\centering}p{2.2cm}>{\centering}p{1.8cm}>{\centering}p{1.2cm}>{\centering}p{1.5cm}}
\hline 
Name & $n$ & $H$ & $EE$ & $\eta$ & $r$\tabularnewline
\hline 
Ants & 74 & 30.998 & 2.64E+02 & 14 & -0.102\tabularnewline
Benguela & 29 & 9.573 & 4.11E+06 & 0 & 0.021\tabularnewline
BridgeBrook & 75 & 56.018 & 9.20E+08 & 48 & -0.668\tabularnewline
Canton & 108 & 40.333 & 3.12E+08 & 24 & -0.226\tabularnewline
CatCortex & 52 & 12.636 & 8.95E+09 & 0 & -0.044\tabularnewline
Centrality\_literature & 118 & 42.976 & 2.44E+08 & 9 & -0.202\tabularnewline
Chesapeake & 33 & 13.240 & 4.71E+02 & 3 & -0.196\tabularnewline
Coachella & 30 & 10.984 & 7.61E+07 & 0 & 0.035\tabularnewline
ColoSpg & 324 & 182.077 & 1.15E+03 & 142 & -0.295\tabularnewline
CorporatePeople & 1586 & 228.395 & 1.27E+10 & 0 & 0.268\tabularnewline
Dolphins & 62 & 20.845 & 2.06E+03 & 2 & -0.044\tabularnewline
Drugs & 616 & 279.467 & 6.91E+07 & 131 & -0.117\tabularnewline
Electronic1 & 122 & 37.694 & 4.84E+02 & 0 & -0.002\tabularnewline
Electronic2 & 252 & 77.982 & 1.04E+03 & 8 & -0.006\tabularnewline
Electronic3 & 512 & 158.658 & 2.17E+03 & 24 & -0.030\tabularnewline
ElVerde & 156 & 51.696 & 4.76E+13 & 5 & -0.174\tabularnewline
Galesburg & 31 & 9.519 & 4.36E+02 & 1 & -0.135\tabularnewline
GD & 249 & 90.440 & 1.60E+04 & 15 & 0.098\tabularnewline
Geom & 3621 & 1462.396 & 4.04E+12 & 537 & 0.168\tabularnewline
Hi\_tech & 33 & 10.975 & 2.95E+03 & 1 & -0.087\tabularnewline
Internet1997 & 3015 & 2148.635 & 6.17E+13 & 1883 & -0.229\tabularnewline
Internet1998 & 3522 & 2473.122 & 1.42E+15 & 2158 & -0.210\tabularnewline
LittleRockA & 181 & 117.772 & 5.32E+17 & 93 & -0.234\tabularnewline
MacaqueVisualCortex & 32 & 9.665 & 1.26E+06 & 1 & 0.008\tabularnewline
Neurons & 280 & 69.083 & 1.31E+10 & 3 & -0.069\tabularnewline
ODLIS & 2898 & 1131.046 & 1.54E+19 & 270 & -0.173\tabularnewline
PIN\_Afulgidus & 32 & 16.366 & 9.91E+01 & 12 & -0.472\tabularnewline
PIN\_Bsubtilis & 84 & 53.144 & 3.52E+02 & 46 & -0.486\tabularnewline
PIN\_Ecoli & 230 & 102.189 & 8.30E+06 & 57 & -0.015\tabularnewline
PIN\_Hpyroli & 710 & 397.649 & 4.60E+04 & 316 & -0.243\tabularnewline
PIN\_KSHV & 50 & 18.119 & 1.82E+03 & 2 & -0.058\tabularnewline
PIN\_Malaria & 229 & 83.377 & 2.25E+04 & 13 & -0.083\tabularnewline
PIN\_Yeast & 2224 & 1135.731 & 1.94E+08 & 754 & -0.105\tabularnewline
Power\_grid & 4941 & 1907.307 & 2.13E+04 & 593 & 0.003\tabularnewline
PRISON & 67 & 20.325 & 7.08E+02 & 0 & 0.103\tabularnewline
ReefSmall & 50 & 12.888 & 2.07E+10 & 0 & -0.193\tabularnewline
Roget & 994 & 264.570 & 2.38E+05 & 2 & 0.174\tabularnewline
Sawmill & 36 & 12.307 & 2.57E+02 & 2 & -0.071\tabularnewline
ScotchBroom & 154 & 103.975 & 2.46E+06 & 90 & -0.311\tabularnewline
Shelf & 81 & 20.724 & 1.60E+18 & 2 & -0.094\tabularnewline
Skipwith & 35 & 15.023 & 3.87E+09 & 7 & -0.319\tabularnewline
SmallWorld & 233 & 115.730 & 1.27E+09 & 70 & -0.303\tabularnewline
College & 32 & 8.049 & 5.36E+02 & 0 & -0.119\tabularnewline
Software\_Abi & 1035 & 575.133 & 1.65E+05 & 418 & -0.086\tabularnewline
Software\_Digital & 150 & 82.277 & 1.31E+03 & 63 & -0.228\tabularnewline
Software\_Mysql & 1480 & 648.971 & 2.70E+09 & 282 & -0.083\tabularnewline
Software\_VTK & 771 & 440.251 & 1.11E+05 & 324 & -0.195\tabularnewline
Software\_XMMS & 971 & 478.168 & 4.64E+04 & 294 & -0.114\tabularnewline
StMarks & 48 & 13.607 & 1.43E+05 & 0 & 0.111\tabularnewline
StMartin & 44 & 14.438 & 2.78E+05 & 2 & -0.153\tabularnewline
Stony & 112 & 41.359 & 7.23E+09 & 30 & -0.222\tabularnewline
Termite\_1 & 507 & 206.581 & 1.92E+03 & 75 & -0.046\tabularnewline
Termite\_2 & 260 & 116.912 & 7.32E+02 & 58 & -0.150\tabularnewline
Termite\_3 & 268 & 100.975 & 1.89E+03 & 23 & 0.045\tabularnewline
Trans\_Ecoli & 328 & 214.517 & 1.06E+04 & 184 & -0.265\tabularnewline
Trans\_urchin & 45 & 22.218 & 9.12E+02 & 13 & -0.207\tabularnewline
Transc\_yeast & 662 & 478.315 & 3.59E+04 & 440 & -0.410\tabularnewline
USAir97 & 332 & 142.765 & 8.08E+17 & 58 & -0.208\tabularnewline
Ythan1 & 134 & 58.374 & 1.86E+07 & 23 & -0.263\tabularnewline
Ythan2 & 92 & 41.326 & 7.07E+06 & 22 & -0.322\tabularnewline
Zackar & 34 & 15.994 & 1.04E+03 & 10 & -0.476\tabularnewline
\hline 
\end{longtable}

\textbf{Table 1. }Dataset of real-world networks studied in this paper,
their size $n$, Gaussian Estrada index $H$, exponential Estrada
index $EE,$ graph nullity $\eta$, and degree assortativity $r$.
\label{Table-1.-Dataset}

\bigskip{}

The largest value of $\hat{H}$ corresponds to the food web of Bridge
Brooks, which displays the second highest normalized nullity. It is
followed by the transcription network of yeast (displaying the highest
value of the normalized nullity) and the versions of Internet at Autonomous
System (AS) of 1997 and 1998. The three networks display triangular
eigenvalue distributions peaked at the zero eigenvalue which explains
their large values of the $\hat{H}$ index. However, while the yeast
transcription network and the Internet at AS have fat-tailed degree
distributions, the Bridge Brooks food web displays a uniform one.
Thus, the existence of large values of the $\hat{H}$ index is not
tied up to the existence of fat-tailed degree distributions. Most
of the networks (75.4\%) have values of the $\hat{H}$ index below
0.5. That is, only 15 networks out of 61 have $\hat{H}\geq0.5$. Among
these 15 networks there are 4 of the 7 protein-protein interaction
networks (PINs) studied and two of the three transcription networks
studied. Thus, almost half of the networks with $\hat{H}\geq0.5$
represent biological systems containing proteomic or transcriptomic
information. The other transcription network studied has $\hat{H}\approx0.494$
and the other 3 PINs have values of $\hat{H}$ ranging between 0.36
and 0.44. It is interesting to explore the main structural causes
for these high values of the $\hat{H}$ index. In previous sections
we have found that the main structural characteristic determining
the high values of this index is the presence of bicliques, e.g. the
highest value of $\hat{H}$ is obtained for complete bipartite graphs,
also the BA networks display larger $\hat{H}$ index that the ER ones
due to the presence of complete bipartite subgraphs created during
the evolution of the preferential attachment mechanism. Consequently,
we should expect that such kind of subgraphs appear in those real-world
networks having the largest $\hat{H}$ index. In the case of the food
web of Bridge Brook we have found a biclique consisting of two sets
of nodes $V_{1}$ and $V_{2}$with cardinalities of 6 and 35 nodes,
respectively. This subgraph represents a biclique $K_{6,35}$ which
contains 55\% of the total number of nodes in the network. There are
also other smaller bicliques in this network, which together with
the $K_{6,35}$ contribute to the large $\hat{H}$ value observed.
In the cases of the yeast transcription network and the Internet at
AS, the networks are characterized by having a few hubs connected
to many nodes of degree one, then producing bicliques of the type
$K_{1,n_{2}}.$ In general these findings can be understood on the
basis of different mechanisms which give rise to the existence of
bicliques in real-world networks. For instance, in some food webs
there are top predators which compete for a group of preys. If for
this group of species there are no prey-prey nor predator-predator
trophic interactions, the corresponding subgraph is a biclique as
the one observed for the Bridge Brook network previously considered.
In the cases of transcription and PINs the bicliques can be formed
as a consequence of lock-and-key kind of interaction. That is, a group
of proteins (genes) can act as locks (activators) that physically
interact with other proteins (activate other genes) acting as keys.
Such kind of interactions is prone to produce relatively large bicliques
in the structure of the networks resulting from them. 

On the other hand, among the networks with $\hat{H}\leq0.3$ we find
the network of corporate directors, the three neuronal networks studied,
i.e., macaque and cat visual cortex and the neuronal network of \textit{C.
elegans}, as well as some social networks and food webs. Also, the
three electronic circuits studied here also display values of $\hat{H}$
index around 0.3. These networks are characterized by the lack of
complete bipartite subgraphs and they may represent a variety of topologies
difficult to be reproduced by a single mechanism. 

Finally we would like to remark a few important characteristics of
the Gaussian matrix function of a network that point out to the necessity
of further studies of it for real-world networks and simple graphs
in general. The first, is our observation that although networks with
fat-tailed degree distribution may give rise to high values of the
$\hat{H}$ index, it is not a necessary condition for a network to
display such a characteristic. We have seen that networks with exponential
and even uniform degree distributions display large values of the
$\hat{H}$ index. Another structural parameter that could be related
to the $\hat{H}$ index is the degree assortativity, i.e., the Pearson
correlation coefficient of the degree-degree distribution of a network.
We have explored such relation between the $\hat{H}$ index and the
assortativity for the 61 networks studied here. We have found that
the two parameters are negatively correlated. That is, high values
of the $\hat{H}$ index in general implies that the networks are disassortative,
i.e., there is a trend of high degree nodes to be connected to low
degree ones. This is understandable on the basis of our findings that
bicliques of the type $K_{1,n_{2}}$ plays a fundamental role in the
value of the $\hat{H}$ index. However, the correlation is very weak
and displays a Pearson correlation coefficient of -0.68. Thus, further
explorations---both theoretical and computational---of the relation
of the $\hat{H}$ index and other network parameters are necessary
for a complete understanding of this index and its application in
network theory.

\section{Summary}

Most of the works using matrix functions for studying graphs are concentrated
on the use of the exponential and the resolvent of the adjacency matrix
of the graph. Other functions such as the hyperbolic sine and cosine,
and $\psi$-matrix functions have also been reported. All these matrix
functions give more weight to the largest eigenvalue and the corresponding
eigenvector of the adjacency matrix than to the rest of eigenvalues/eigenvectors.
In many real-world networks, where the spectral gap is relatively
large, this situation gives rise to discarding important structural
information contained in the eigenvalues close to zero in the graph
spectra. Here, we have studied a Gaussian matrix function which accounts
for the information contained in the eigenvalues/eigenvectors close
to zero in the graph spectra. We have shown that such information
is related to the existence of important structural patterns in graphs
which have remained unexplored when studying the structure of complex
networks, such as the existence of relatively large complete bipartite
subgraphs (bicliques). Such bicliques appear naturally in many real-world
networks as well as in the Barab�si-Albert graphs and other networks
with fat-tailed degree distributions. In this work we have concentrated
in the theoretical characterization of the networks displaying the
largest Gaussian Estrada index---an index characterizing the importance
of eigenvalues close to zero. Other extensions to give more weight
to other specific eigenvalues/eigenvectors of the adjacency matrix
are under development. We hope this work will open new research interest
in the study of matrix functions for the structural characterization
of graphs.

\end{document}